\pdfoutput=1
\documentclass{llncs}
\usepackage[usenames, dvipsnames]{xcolor}
\usepackage[utf8]{inputenc}
\usepackage[english]{babel}
\usepackage{amsmath,amsfonts}
\usepackage[textwidth=1.1in]{todonotes}
\usepackage{cite}
\usepackage{mathtools}
\usepackage{tikz}
\usepackage[a4paper, portrait, margin=1.3in, top=1.2in, bottom=1in]{geometry}
\usepackage[numbers,sectionbib]{natbib}
\usepackage[hyperindex,breaklinks]{hyperref}

\usetikzlibrary{patterns}

\newcommand\cupdot{\mathrel{\ooalign{$\cup$\cr\hidewidth$\cdot$\hidewidth\cr}}}
\newcommand\polylog{\mathop{\textup{polylog}}}
\newcommand\eqdef{\mathrel{\overset{\makebox[0pt]{\mbox{\normalfont\tiny\sffamily def}}}{=}}}
\DeclarePairedDelimiter{\ceil}{\lceil}{\rceil}
\def\squareforqed{\leavevmode\hbox to.77778em{\hfil\vrule\vbox to.675em{\hrule width.6em\vfil\hrule}\vrule\hfil}} 

\spnewtheorem{open}{Open problem}{\bfseries}{\itshape}
\spnewtheorem{problem2}{Problem}{\bfseries}{}
\spnewtheorem{observation}{Observation}{\bfseries}{\itshape}

\title{Theoretical Model of Computation and Algorithms for FPGA-based Hardware Accelerators\thanks{This work was carried out while the authors were participants in the DIMATIA-DIMACS REU exchange program at Rutgers University.\hfill\break The work was supported by the grant SVV–2017–260452.}}
\author{Martin Hora\inst{1}
   \and V\'{a}clav Kon\v{c}ick\'{y}\inst{2}
   \and Jakub T\v{e}tek\inst{2}
}
\authorrunning{M. Hora \and V. Kon\v{c}ick\'{y} \and J. T\v{e}tek}

\institute{Computer Science Institute, Charles University,\\ Malostransk\'{e} n\'{a}m. 25, CZ -- 11800 Prague, Czech Republic\\
           \email{mhora@iuuk.mff.cuni.cz}
      \and Department of Applied Mathematics, Charles University,\\ Malostransk\'{e} n\'{a}m. 25, CZ -- 11800 Prague, Czech Republic\\
           \email{\{koncicky,jtetek\}@kam.mff.cuni.cz}
}
\begin{document}

\maketitle

\begin{abstract}
While FPGAs have been used extensively as hardware accelerators in industrial computation \cite{cit:uses_FPGA}, no theoretical model of computation has been devised for the study of FPGA-based accelerators. In this paper, we present a theoretical model of computation on a system with conventional CPU and an FPGA, based on word-RAM. We show several algorithms in this model which are asymptotically faster than their word-RAM counterparts. Specifically, we show an algorithm for sorting, evaluation of associative operation and general techniques for speeding up some recursive algorithms and some dynamic programs. We also derive lower bounds on the running times needed to solve some problems.
\end{abstract}

\section{Introduction}
While Moore's law has dictated the increase in processing power of processors for several decades, the current technology is hitting its limits and the single core performance is largely staying the same. A new paradigm is needed to speed up computation. Many alternative architectures have been made, such as multicore, manycore (specifically GPU) and FPGA, which is the one we consider in this paper. FPGA (Field-programmable gate array) is a special hardware for efficient execution of boolean circuits. While there exists a theoretical model which allows theoretical treatment of algorithms in the multicore and manycore model \cite{cit:model_GPU}, no such model has been devised for computation with FPGA-based accelerators, despite their rising popularity. In this paper, we consider heterogeneous systems which have a conventional CPU and an FPGA. This is especially relevant since Intel has recently introduced combined CPU with an FPGA \cite{cit:intel_combined_xeon}, which should lower the latency between CPU and FPGA. We define a theoretical model for such systems and show several asymptotically fast algorithms on this model solving some fundamental problems. 

\subsection{Previous work}
While there is no theory of FPGA algorithms, there are experimental results which show that FPGA can be used to speed up the computation by one to two orders of magnitude and decrease the power consumption by up to three orders of magnitude. \citet{cit:bio_paper} and \citet{cit:bio_thesis} show how FPGA can be used to speed up computations in bioinformatics. FPGAs have also been used to speed up genetic algorithms, as shown by \citet{cit:genet_alg}. In the following paper, the authors discuss what causes the performance gains of using an FPGA \cite{cit:why_faster}. \citet{cit:compare_GPU} and \citet{cit:GPU_compare_sort} showed a comparison with GPU and multicore CPU. Several algorithms have been implemented on FPGA and benchmarked, including matrix inversion \cite{cit:FPGA_matrix_inv}, AES \cite{cit:FPGA_AES} and k-means \cite{cit:FPGA_k-means}.

\subsection{Our contributions}
The main contribution of this paper is the definition of a model of computation which captures computation with FPGA-based accelerators (Section~\ref{section:model}). In this model, we speed up classical algorithms for some fundamental problems, specifically sorting (Section~\ref{section:sorting}), evaluation of associative operation (Section~\ref{section:assoc_op}), speed up some dynamic programming problems including longest common subsequence and edit distance (Section~\ref{section:dynprog}) and speed up some recursive algorithms including Strassen algorithm, Karatsuba algorithm as well as many state-of-the-art algorithms for solving \textsf{NP}-hard problems (Section~\ref{section:recursive}). In Section~\ref{section:lower_bounds} we show some lower bounds which follow from a relation to cache-aware algorithms. 

\section{Preliminaries}

\subsection{Boolean Circuits}
\begin{definition}[Boolean circuit]
A boolean circuit with $n$ inputs and $m$ outputs is a labeled directed acyclic graph. It has $n$ vertices $x_1, \cdots, x_n$ with no incoming edges, which we call the input nodes, and $m$ vertices $y_1, \cdots, y_m$ with no outgoing edges, which we call output nodes. We call the vertices $x_1, \cdots x_n, y_1, \cdots,\allowbreak y_m$ the I/O nodes. Each non-input vertex in this graph corresponds to a boolean function from the set $\{\vee, \wedge, \neg, \text{id}\}$ and is called a gate. All gate vertices have indegree equal to the arity of their boolean function.

The circuit computes a function $\{0,1\}^n \to \{0,1\}^m$ by evaluating this graph. At the beginning of computation, an input node $x_i$ is set to the $i$-th bit of input. In each step of computation, gates which received their input calculate the function and send their output to the connected vertices. When all output nodes receive their value, the circuit outputs the result with the value of $y_i$ being the $i$-th bit of output.
\end{definition}

The order of a circuit is the number of its input nodes. Note that in section \ref{section:recursive} we introduce a different notion of order of a circuit, which does not equal the number of circuit's input nodes.

We consider two complexity measures of circuits -- \emph{circuit size}, defined as the number of gates, and \emph{depth}, defined as the length of the longest path from an input node to an output node.

We denote the size and depth of a circuit $C$ by $\text{size}(C)$ and $\text{del}(C)$, respectively. Size complexity and depth of a family of circuits $\mathcal{C}$ is denoted by $\text{size}_{\mathcal{C}}(n)$ and $\text{del}_{\mathcal{C}}(n)$, respectively, where $n$ is the order of the circuit.

\begin{definition}[Synchronous circuit]
A circuit is synchronous if for every gate $g$, all paths from input nodes to $g$ have the same length and all output nodes are in the same depth.
\end{definition}

While any circuit can be made synchronous by adding intermediate identity gates, this can asymptotically increase the number of gates used. Note that our notion of synchronous circuits is stronger than the one introduced by \citet{cit:synch_circ}. The original definition does not require the output gates to be in the same depth.

\medskip
\noindent For more detailed treatment of this topic, refer to \cite{cit:circuit_book}.

\subsection{Field Programmable Gate Arrays} \label{section:fpga}
Field Programmable Gate Array, abbreviated as FPGA, is hardware which can efficiently execute boolean circuits. It has programmable gates, which can be programmed to perform a specific function. The gates can then be connected via interconnects into the desired circuit. Usual FPGA consists of a grid of gates with interconnects between them. The architecture of FPGAs only allows for the execution of synchronous circuits, but multiple synchronous circuits of different depths can be realised on an FPGA. There are techniques for delaying signals which can be used for execution of asynchronous circuits. However, only limited number of signals can be delayed. For this reason, we limit the model to synchronous circuits.

Let $G$ denote the number of gates of the FPGA. Then we assume that any set of circuits with a total of at most $G$ gates can be realised by the interconnects. This assumption tends to be mostly correct in practice, at least asymptotically. The FPGA communicates with other hardware using I/O pins, which are usually placed at the edges of the FPGA. Let $I$ denote the number of I/O pins. It usually holds that $I \approx \sqrt{G}$ since the I/O pins are placed at the edges of the chip.

In the theoretical model, we count execution of a circuit on FPGA as a single operation, not accounting for the memory transfers. This is to capture the increased throughput of FPGA compared to a conventional CPU. Moreover, the speed of RAM is on the same order of magnitude as the speed with which the data can be read by the FPGA. For example, Intel claims that their new FPGA embedded on CPU has I/O throughput of 160Gbs \cite{cit:intel_combined_xeon}, while modern RAM memories have both reading and writing speed at least 160Gbs (for dual-channel) and up to 480Gbs (for quad-channel) \cite{cit:RAM_speed}.

The computation of the FPGA takes place in discrete time-steps and can be pipelined. Each time-step, output signals of layer $i$ of the circuit go to the layer $i+1$, while layer $i$ receives and processes the output signals of layer $i-1$. This means that in each time-step, we can give the FPGA an input and we get the output delayed by a number of steps proportional to the depth of the circuit. If we run the FPGA at least as many times as is the depth of the circuit, the amortized time spent per input is, therefore, $\mathcal{O}(1)$.

Reprogramming an FPGA can take up a significant amount of time. For this reason, we require in the model that the circuits used in an algorithm are constructed beforehand and stay fixed during the computation.

\section{Model of Computation} \label{section:model}

Throughout this paper, we use the following model of computation. We call the model Pipelining Circuit RAM, abbreviated as PCRAM. The computer consists of word-RAM of word-size $w$ together with a circuit execution module with $G$ circuit gates and $I$ circuit inputs/outputs. A program in the PCRAM model consists of two parts. First is a program in the word-RAM running in time polynomial in $G$ and $I$, taking no input and outputting a sequence of circuits $\mathcal{C} = C_1, \cdots, C_m$ for some $m$, in a standard adjacency list representation (i.e. for each gate, the incoming edges are specified) with a total of at most $G$ gates and $I$ I/O nodes. Since the algorithm generating the circuits does not take any input except $G$ and $I$, the sequence $\mathcal{C}$ only depends on this parameters. The time needed for generating $\mathcal{C}$ does not count towards the time complexity of the algorithm. After the end of the first phase, the sequence of circuits $\mathcal{C}$ cannot be changed. Second part is the main program, working in a modified word-RAM:

At the beginning of execution of the main program, the values $G$ and $I$ are stored in memory. We always assume that $G \ge I$. Additionally, one might assume that $G \in \Theta(I^2)$ because it is usually the case in practice as explained in Section~\ref{section:fpga}. This assumption can greatly simplify sharing of resources, as is described below.

The word-RAM has an additional instruction $\textsc{RunCircuit}(i, s, t)$ that starts computation of circuit $C_i$ on specified input. The instruction has three parameters. Parameter $i \in [m]$ identifies which circuit should be run. Parameter $s$ is the source address of the input data. The address $s$ has to be aligned to whole words. If the circuit has $l$ inputs then the contiguous block of $l$ bits starting from address $s$ is used as input for the circuit. Similarly, $t$ is the address where the output of the circuit is stored after the circuit computes its output. The output is also stored as a contiguous block of bits and, as with the input, the address $t$ has to be aligned to whole words. However, note that variable shift operation by at most $w$ bits can be easily implemented in depth $\mathcal{O}(\log \log w)$, making it possible to use unaligned inputs and outputs.
Instruction \textsc{RunCircuit} lasts one time-step and it only starts execution of the circuit and does not wait for the results. 
The computation time of a circuit is proportional to its depth. If we are evaluating instruction $\textsc{RunCircuit}(i, s, t)$ then we get the output starting at address $t$ exactly $\text{del}(C_i)$ steps after we executed the instruction. This can result in concurrent writes. While concurrent writes are allowed, the resulting value is undefined and the memory cells which are in the intersection of the writes can have any value.
The model allows for use of the \textsc{RunCircuit} instruction without waiting for the results of the previous \textsc{RunCircuit} instruction. This is called pipelining and it is a critical feature of the PCRAM model which makes it possible to speed up many algorithms.

We also consider a randomized version of PCRAM with an extra operation which assigns to a memory cell a value chosen uniformly at random independently from other calls of the operation.

As in the case of (word) RAM, we measure the time complexity by the number of operations executed in a worst case. We express the complexity of an algorithm in terms of the input size $n$, the word size $w$, the number of available circuit inputs and outputs $I$ and the number of available gates $G$. The complexity of a query on a data structure is measured in time, defined as the number of steps executed on the RAM, and delay, defined as the number of steps when the RAM is waiting for an output of a circuit. The reason for devising two complexity measures is that when performing multiple queries, the processor time is sequential and cannot be pipelined, whereas the execution of the circuit can potentially be pipelined, not necessarily resulting in the delay adding up.

\medskip \noindent
Note that in general, it is not possible to use multiple circuits and assume that each can use all $G$ gates and $I$ I/O nodes. However, if all circuits use at most polynomial number of gates in their number of input/output nodes, the speedup caused by the circuits is at most polynomial in the number of their I/O nodes. Shrinking the circuits by a constant factor (that is, taking circuits for values of $I$ and $G$ smaller by a constant factor) will, therefore, incur at most constant slowdown. In this paper, whenever we use multiple circuits in an algorithm, they all use a polynomial number of gates and we do not further discuss the necessity to share resources between circuits.

We assume that we can copy $\mathcal{O}(I)$ bits in memory in time $\mathcal{O}(1)$. We call this the \textsc{Copy} instruction. For the case when $G$ is polynomial in $I$, this follows from the following theorem (however, in practice it would not make sense to use FPGA for copying data).

\begin{theorem} \label{th:copy}
	In the PCRAM model, if $G \in \mathcal{O}(\text{poly}(I))$, a contiguous block of $\mathcal{O}(I)$ bits can be copied in time $\mathcal{O}(1)$ without use of the \textsc{Copy} instruction.
\end{theorem}
\begin{proof}
	Let $CP_k$ be a copy circuit of size $k$ that is formed by $k$ identity gates. $CP_k$ copies a contiguous block of $k$ bits in time $\mathcal{O}(1)$ and it uses $2k$ I/O pins. 
	Next, let $K$ denote the greatest power of two such that $8K \le I$. We put $\mathcal{C} = \{ CP_1, CP_2, CP_4, \dots, CP_K\}$. In total the circuits in $\mathcal{C}$ use $2K - 1 < G$ gates and $4K - 2 < I$ I/O pins.
	
	For every contiguous block of $\mathcal{O}(I)$ bits there exists $s \in \{ 2^0,2^1,2^2,\dots, K\}$ such that the block can be covered by $\mathcal{O}(1)$ contiguous intervals of size $s$, not necessarily pairwise disjoint. The block can then be copied by $\mathcal{O}(1)$ calls of $\textsc{RunCircuit}(CP_s)$ instruction in time $\mathcal{O}(1)$.

At most one half of available gates and I/O gates are used for the \textsc{Copy} circuit. We set $G' = \frac{1}{2}G$ and $I' = \frac{1}{2}I$ as the parameters of the circuit module for the other circuits. Since $G \in \mathcal{O}(\text{poly}(I))$, we incur at most constant slowdown by a constant decrease in the model's parameters.
\qed
\end{proof}
 
The running time of the algorithms will often depend on ``number of elements that can be processed on the circuit at once''. We usually call this number $k$ and define it when needed.

\section{Algorithmic Results}
In this section, we show several algorithms that are asymptotically faster than the best known algorithms in the word-RAM.

It would also be possible to use in PCRAM a circuit to speed up data structures. However, since efficient data structures often depend on doing random access in memory, it will be difficult to improve the running time of the data structure by more than a factor of $\Theta(\log(I))$, as speedups gained from circuits depend on sequential access to data. It is likely that this improvement will not be enough to justify practical use of an FPGA. In Section~\ref{section:lower_bounds}, we give lower bounds in the PCRAM model, including lower bound on search trees.

\subsection{Aggregation}\label{section:aggregation}
We show an efficient algorithm performing aggregation operation on an array. While this operation cannot be done in constant time on word-RAM, pipelining can be used to get constant time per operation when amortized over great enough number of instances in the PCRAM model. In Section~\ref{section:sorting}, we use this operation for sorting numbers.

\begin{problem2}[Aggregation]\label{problem:aggregation}
\strut \\
\emph{Input:} array $A$ of $n$ numbers, $n$-bit mask $M$ \\
\emph{Output:} $t$ -- the number of ones in $M$ and array $B$ -- a permutation of $A$ such that if $M[i] = 1$ and $M[j] = 0$, then $A[i]$ precedes $A[j]$ in array $B$ for all $i,j$.
\end{problem2}

In this section, we use ${k \in \Theta\!\left(\min\left(\frac{I}{w},\frac{G}{w} / \log^2\frac{G}{w}\right)\right)}$.

\begin{theorem}\label{th:aggregation}
	Aggregation of $n$ numbers can be computed on PCRAM in time $\mathcal{O}\!\left(\frac{n}{k} + \log^2 k\right)$.
\end{theorem}
\begin{proof}
	Without loss of generality, we suppose that $n$ is divisible by $k$. Otherwise, we use a circuit to extend the input to the smallest greater size divisible by $k$.

	We construct a circuit that can aggregate $k$ numbers using a sorting network of order $k$ \cite{cit:bitonic}. We call this the aggregator circuit. We use a sorting network to sort $M$ in descending order. The sorting network uses modified comparators that do not only compare and potentially swap values in $M$, but also swap the corresponding items in array $A$ when swapping values in $M$.
	We use a synchronous sorting network called the bitonic sorter \cite{cit:bitonic}. Bitonic sorter of order $k$ uses $\mathcal{O}(k \log^2 k)$ comparators and has depth $\mathcal{O}(\log^2 k)$.
	The modified comparator that also moves items in array $A$ can be implemented using $\mathcal{O}(w)$ gates in depth $\mathcal{O}(1)$, since we are only comparing one-bit values. Therefore, our modification of the sorting network uses $\mathcal{O}(kw \log^2 k)$ gates, its depth is $\mathcal{O}(\log^2 k)$ and uses $\mathcal{O}(kw)$ I/O nodes.

	Number $t$ has $\ceil{\log k}$ bits and each of them can be derived independently on others from the position of the last occurrence of bit $1$ in sorted array $M$. 
	First, the circuit determines in parallel for every position of sorted $M$ whether it is the last occurrence of $1$. This can be done in constant depth with $\mathcal{O}(k)$ gates. 
	Then, for every bit of $t$, the circuit checks if the last $1$ is at one of the positions that imply that the bit of $t$ is equal to $1$. The value of one bit of $t$ can be computed in depth $\mathcal{O}(\log k)$ with $\mathcal{O}(k)$ gates using a binary tree of OR gates, therefore we need $\mathcal{O}(k \log k)$ gates to compute the value of $t$ from sorted $M$ in depth $\mathcal{O}(\log k)$.
	
	In total the circuit uses $\mathcal{O}(kw \log^2 k)$ gates, it takes ${k(w+1)}$ input bits, produces ${kw + \ceil{\log k}}$ output bits and its depth is $\mathcal{O}(\log^2 k)$.

\medskip\noindent
	The algorithm runs in two phases:

	In the \emph{first phase}, split the input into chunks $(A_1, M_1), (A_2, M_2), \dots (A_{n/k},\allowbreak M_{n/k})$ of size $k$ and use the aggregator circuit to separately aggregate elements in each block. The circuit produces outputs $t_1, t_2, \dots, t_{n/k}$ and $B_1, B_2, \dots, B_{n/k}$. It may be necessary to wait for the delay after this phase.
	
	In the \emph{second phase}, for every $i \in {1, 2, \dots, n/k}$ split array $B_i$ into two arrays $P_i$ and $S_i$, where $P_i$ is the prefix of $B_i$ of length $t_i$ and $S_i$ is the rest of $B_i$. Concatenate the arrays $P_1, P_2, \dots, P_{n/k}, S_1, S_2, \dots, S_{n/k}$ using the \textsc{Copy} instruction.
	
\smallskip\noindent
	In total, the algorithm runs in time $\mathcal{O}\!\left(\frac{n}{k} + \log^2 k\right)$. \qed
\end{proof}

The construction can likely be asymptotically improved by using a sorting network with $\mathcal{O}(n \log n)$ comparators and depth $\mathcal{O}(\log n)$ \cite{cit:opt_sort_netw}. However, such sorting networks are complicated and impractical due to large multiplicative constants.

\begin{theorem}\label{th:bulk:aggregation}
	The aggregation operation on $m$ arrays with a total length $n$ can be computed in time $\mathcal{O}\big(\frac{n}{k} + m + \log^2 k\big)$.
\end{theorem}   
\begin{proof}
	We use the algorithm from the proof of Theorem~\ref{th:aggregation} and interleave its executions to decrease delay.

	First, we run the first phase of the algorithm from proof of Theorem~\ref{th:aggregation} for all the arrays and then run the second phases. This way, we only wait for the delay before the second phase at most once. \qed
\end{proof}

\subsection{Sorting} \label{section:sorting}
We show an asymptotically faster modification of randomized Quicksort with expected running time $\mathcal{O}\left(\frac{n}{k}\log n + \polylog(n,k,w)\right)$. We achieve this by improving the time complexity of the partition subroutine. Throughout this section, we use ${k \in \Theta\big(\min\big(\frac{I}{w},\frac{G}{w \log w} / \log^2 \frac{G}{w \log w}\big)\big)}$.

\begin{problem2}[Pivot Partition]\label{problem:pivot}
\strut \\
\emph{Input:} array $A$ of $n$ numbers, pivot $p$ \\
\emph{Output:} arrays $A_1$, $A_2$ and the integer $|A_1|$, such that $A_1 \cupdot A_2 = A$ and $A_1$ consists exactly of elements $a \in A,\text{ s.t. }a \leq p$, where $\cupdot$ is the disjoint union.
\end{problem2}

\begin{theorem}\label{th:pivot}
	Pivot partition of $n$ numbers can be computed in the PCRAM model in time $\mathcal{O}(n/k + \log^2 k + \log w)$.
\end{theorem}
\begin{proof}
	We reduce the pivot partition problem to aggregation. For each $0 \le i < n$ we set $M[i] = [A[i] \le p]$ and perform aggregation with input $(A, M)$, getting output $(t, B)$. $A_1$ consists of the first $t$ elements of $B$ and $A_2$ of the rest of $B$. The array length $|A_1|$ is then equal to $t$.
	
	The reduction can be done on PCRAM in time $\mathcal{O}(n/k + \log w)$ by using a circuit to compute $k$ values of array $M$ at once in depth $\mathcal{O}(\log w)$ with $\mathcal{O}(kw)$ gates. Together with aggregation, this gives the desired running time. \qed
\end{proof}

Similarly to aggregation, pivot partition problems can be solved in bulk. The next theorem follows from Theorem~\ref{th:bulk:aggregation} by the same reduction as in Theorem~\ref{th:pivot}.

\begin{theorem}\label{th:pivot:bulk}
	We can solve $m$ independent pivot partition problems for arrays of total length $n$ in the PCRAM model in time $\mathcal{O}(n/k + m + \log^2 k + \log w)$.
\end{theorem}

\begin{theorem}
	There is a randomized algorithm in the PCRAM model which sorts $n$ numbers in expected time $\mathcal{O}\left(\frac{n}{k}\log n + \polylog(n,k,w)\right)$.
\end{theorem}
\begin{proof}
We show a modification of randomized Quicksort. We use bulk pivot partitioning from Theorem~\ref{th:pivot:bulk} and a bitonic sorter \cite{cit:bitonic} of order $k$ sorting $w$-bit numbers. Each comparator can be synchronously implemented by $\mathcal{O}(w \log w)$ gates in depth $\mathcal{O}(\log w)$ by taking a standard comparator of depth $\mathcal{O}(\log w)$ and adding indentity gates to make it synchronous. Therefore, the bitonic sorter has $\mathcal{O}(kw \log^2{k}\log{w})$ gates and depth $\mathcal{O}(\log^2{k}\log{w})$. Note that we can also use the bitonic sorter of order $k$ to sort less than $k$ numbers if we fill the unused inputs by the maximum $w$-bit value and it is therefore enough to only have sorting circuit of one size.	
	The algorithm traverses the recursive tree of Quicksort algorithm in BFS manner. Whenever a subproblem has size at most $k$, it is sorted by the bitonic sorter.
	Each layer of the recursion tree corresponds to a sequence of blocks $B_1, \dots, B_m$ for some integer $m$ to be sorted satisfying that if $b_1 \in B_j, b_2 \in B_\ell$, such that $1 \le j < \ell \le m$, then $b_1 \leq b_2$.
	In each layer, blocks of length at most $k$ are sorted using the bitonic sorter and the remaining blocks are partitioned into two by the bulk pivot partition algorithm from Theorem~\ref{th:pivot:bulk} with pivots being chosen uniformly at random for each of the blocks.
	
	In each layer, there can be at most $n/k$ blocks with more than $k$ elements. Thus, from Theorem~\ref{th:pivot:bulk} follows that the partition requires $\mathcal{O}(n/k + \log^2 k + \log w)$ time. Furthermore, there can be at most $2n/k$ blocks of size at most $k$. We do not have to wait for the delay of the sorting circuit before processing of the next layer. Therefore, we spend $\mathcal{O}(n/k + \log^2 k + \log w)$ time per layer.
	
	The expected depth of the Quicksort recursion tree is $\mathcal{O}(\log n)$ \cite{cit:rand_tree_height} (note that Quicksort recursion depth is the same as the depth of a random binary search tree, which is proven to be $\mathcal{O}(\log n)$ in the cited paper). It follows that the total expected time complexity of our algorithm is $\mathcal{O}\big(\big(n/k + \log^2 k + \log w\big) \log n + \log^2 k \log w \big)$. \qed
\end{proof}

\subsection{Associative operation}\label{section:assoc_op}
Let $\otimes : \{0, 1\}^w \times \{0, 1\}^w \rightarrow \{0, 1\}^w$ be an associative operation (such as maximum or addition over a finite group). Let $C_1$ be a synchronous boolean circuit with $g'$ gates, $2w$ inputs and depth $d'$ which computes the function $\otimes$. We want to efficiently compute $\bigotimes_{i=1}^n a_i$ for input numbers $a_1, \dots, a_n \in \{0, 1\}^w$.

We can assume without loss of generality that there is a neutral element $e$ satisfying $(\forall{x \in \{0, 1\}^w})\;(x \otimes e = e \otimes x = x)$. If there is no such $e$, we extend $\otimes$ to have one. This only changes the size and depth of $C_1$ by a constant factor.

\begin{problem2}[$\bigotimes$-sum]\label{problem:sum}
\strut \\
\emph{Input:} array $A$ of $n$ numbers $a_1, \dots, a_n \in \{0, 1\}^w$ \\
\emph{Output:} value $\bigotimes_{i=0}^{n-1} A\left[i\right]$
\end{problem2}

The idea of the algorithm is to split the input into blocks that can be processed by the circuit simultaneously by pipelining the computation. In order to do this, we arrange copies of circuit $C_1$ into a full binary tree. For $j \ge 1$ we inductively define circuit $C_{j+1}$ that consists of two copies of $C_j$ whose outputs go into a copy of $C_1$.

In this section, we use $k \in \Theta\big(\min\big(\frac{I}{w}, \frac{G}{g'}\big)\big)$. We assume that $k$ is a power of two.

\begin{theorem}
$\bigotimes$-sum of $n$ $w$-bit numbers can be evaluated on PCRAM in time $\mathcal{O}(\frac{n}{k} + d'\log{(kd')})$.
\end{theorem} 
\begin{proof}
The algorithm repeatedly calls subroutine \textsc{SumBlocks} which splits the input into blocks of size $k$ and computes the $\otimes$-sum of the block using the circuit $C_{\log k}$. The results from $C_{\log k}$ are collected in a buffer that is used as the input for the next call of the subroutine. Each call of \textsc{SumBlocks} reduces the input size $k$ times, therefore we have to call the subroutine $\left\lceil\log_k n\right\rceil$ times.

If the length of the input of subroutine \textsc{SumBlocks} is not divisible by $k$ we fill the last block to the full size with neutral element $e$. By an argument similar to that in Theorem~\ref{th:copy}, this can be done in constant time by having $\Theta(\log n)$ circuits, the $i$-th of which has size $w2^i$, has no inputs and produces $2^i$ copies of $e$.

For an input of length $n'$, \textsc{SumBlocks} invokes $\big\lceil\frac{n'}{k}\big\rceil$ calls of $C_{\log k}$ and runs in time $\Theta\big(\frac{n'}{k}\big)$ with delay $\Theta(d'\log k)$.

The size of the problem decreases exponentially, so the first call of \textsc{SumBlocks} dominates in terms of time complexity. The delay is the same for each call of \textsc{SumBlocks}. However, if the input has length of $\Omega(kd'\log k)$ words, the individual values are computed by the circuit before they are needed in the current call of \textsc{SumBlocks} and it is, therefore, not necessary to wait for the delay. We only wait for the delay the last $\mathcal{O}(\log_k (kd'\log k)) = \mathcal{O}(\log_k (kd'))$ executions of \textsc{SumBlocks}, resulting in total waiting time of $\mathcal{O}(d' \log k \log_k (kd')) = \mathcal{O}(d' \log (kd'))$.

\end{proof}

\subsection{Dynamic programs} \label{section:dynprog}

\paragraph{A dynamic program} is a recursive algorithm where the computed values are being memoized, meaning that the function is not computed if it has already been computed with the same parameters. Many dynamic programs used in real-world applications have the property that the subproblems form a multidimensional grid. We then call the individual subproblems cells. For simplicity, we will focus only on dynamic programs with a rectangular grid. However, it is possible to generalise the described technique to higher dimensions. 

Many dynamic programs such as those solving the longest common subsequence, shortest common supersequence and edit distance problems or the Needleman-Wunsch algorithm satisfy the following three properties:
\begin{enumerate}
\item Each cell depends only on a constant number of other cells.
\item Each cell $C$ depends only on cells which are in the upper-left quadrant with $C$ in the corner (as shown in figure~\ref{fig:culd_dyn}).
\item For any cell $C$, the cells that $C$ depends on have the same relative position with respect to $C$ (as shown in figure~\ref{fig:dyn_dependency}).
\end{enumerate}
We call such dynamic programs CULD (constant upper-left dependency). 

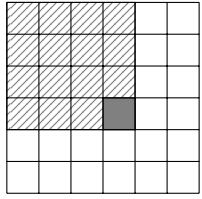
\begin{figure}[h]
\centering\begin{tikzpicture}[x=12pt,y=12pt]
{
\draw[pattern=north east lines, pattern color=gray] (0,6) rectangle (4,2);
\draw[fill=gray] (3,2) rectangle +(1,1);
\draw[step=1] (0,0) grid (6,6);
}
\end{tikzpicture}
\caption{The grey cell can only depend on the hatched cells in a CULD dynamic program}
\label{fig:culd_dyn}
\end{figure}

A CULD dynamic program has linear number of cells that are near the left or top edge of the grid and, therefore, should have a dependency on a cell that does not exist (because it would have at least one coordinate negative). These cells have to have their values computed in some other way and are called base cases. As it is usually trivial to compute the values of these cells, we will not deal with computing them in this paper.

Dynamic programs can be sped up on the PCRAM model using the circuit for (1) speeding up the computation of the individual values (it has to take more than constant time to compute each value for this to be possible) or (2) parallelizing the computation of the values (usually when the individual values can be computed in constant time). Computation of individual values of the function can be sped up in the special case of associative operation using the algorithm described in Section~\ref{section:assoc_op}. In this section, we show how to use the second approach to speed up the computation of the CULD dynamic programs.

\begin{theorem}
Let $P$ be a CULD dynamic program with dependency grid of size $M \times N$. Let $g'$ and $d'$ be the number of gates and delay, respectively, of a circuit which computes value of a cell from the values that the cell depends on. After all base cases of the dynamic program have been computed, the remaining values in the dependency grid can be computed in time $\mathcal{O}\big(\frac{MN}{k} + M + N + kd'^2\big)$ where $k \eqdef \min\big(\frac{I}{w}, \frac{G}{g'}\big)$.
\end{theorem}
\begin{proof}
Instead of storing the grid in a row-major or column-major layout, we store it in antidiagonal-major layout. That is, the antidiagonals are stored contiguously in memory.

Let $d$ be the number of cells that a cell depends on. We calculate the values of $k$ cells in parallel using a circuit. Notice that subsequent cells in the layout depend on at most $d$ contiguous sequences of cells (see figure~\ref{fig:dyn_dependency}). First, we move these cells to one contiguous chunk of memory. This can be used as input to a circuit which computes the values of the $k$ cells.

This process can be done repeatedly for all antidiagonals, computing values of all cells in the grid. Note that it is necessary to wait for the delay only if the antidiagonal that is being computed has $\mathcal{O}(kd')$ cells. The total waiting time is, therefore, $\mathcal{O}(kd'^2)$ and the total time complexity $\mathcal{O}\big(\frac{MN}{k} + M + N + kd'^2\big)$. \qed
\end{proof}
The circuit computing the cell value can be obtained either by using the general technique of simulating a RAM, which we show in Section~\ref{simulation}, or an ad-hoc circuit can be used.

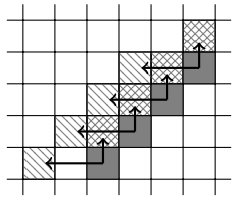
\begin{figure}[h]
\centering\begin{tikzpicture}[x=12pt,y=12pt]
{
\def\cell(#1,#2);{
\draw[fill=gray](#1,#2) rectangle +(1,1);
\draw[pattern=crosshatch, pattern color=gray] (#1,#2)++(0,1) rectangle +(1,1);
\draw[pattern=north west lines, pattern color=gray] (#1,#2)++(-2,0) rectangle +(1,1);
\draw[thick,->] (#1,#2)++(0.5,0.5) -- ++(0,0.8);
\draw[thick,->] (#1,#2)++(0.5,0.5) -- ++(-1.8,0);
}
\draw[step=1] (-1.5,0.5) grid (5.5,6.5);
\cell(1,1);
\cell(2,2);
\cell(3,3);
\cell(4,4);
}
\end{tikzpicture}
\caption{Contiguous sequence of cells (grey) depend on a constant number (in this case two) of contiguous sequences of cells (hatched)}
\label{fig:dyn_dependency}
\end{figure}
\subsection{Recursive algorithms} \label{section:recursive}
In this section, we describe how the circuit can be used to speed up recursive algorithms in which the amount of work is increasing exponentially in the recursion level. This includes some algorithms whose time complexity can be analyzed by the Master theorem \cite{cit:alg_intro} as well as many exponential time algorithms. Our approach does not result in a substantial increase of memory usage. 

We use the circuit to solve the problem for small input sizes. We do this by using the concept of totally computing circuits, which we introduce. Totally computing circuit is a circuit that computes a given family of functions for all inputs of length bounded by some constant. This is a natural concept in PCRAM since we often want to use the circuit to compute a function on arbitrary input that is short enough and not only for inputs of a specific length.

\begin{definition}[Totally computing circuit] \label{def:total}
Let $\{f_k\}_{k\in \mathbb{N}}$ be a family of functions such that $f_k : \{0, 1\}^k \rightarrow \{0, 1\}^k$ for every $k$. 
Boolean circuit $C_n$ is a totally computing circuit of order $n$ of sequence $\{f_k\}_{k\in \mathbb{N}}$ if it has $n + \ceil{\log_2 n}$ inputs, interpreting the first $\ceil{\log_2 n}$ bits of input as a binary representation of number $n'$ and it outputs on the first $k$ outputs the value of function $f_{n'}$ evaluated on the next $n'$ bits of the input.
\end{definition}

A similar definition could be made, where the circuit would work in ternary logic with one of the values meaning ``no input'', requiring that the input is contiguous and starts at the first I/O node. We stick with the definition mentioned above because this notion of totally computing circuits can be easily simulated by totally computing circuits according to Definition~\ref{def:total}.

\begin{definition}[$(f,G)$-recursive algorithm with independent branches]
We call an algorithm $A$ in word-RAM model to be $(f,G)$-recursive with independent branches if it is of the form $A(0) = 0$ (the recursion base case) and $A(x) = f(A(g_1(x)), \dots, A(g_k(x)))$ otherwise, where the functions $f$ and $g_i$ have no side effects, can use randomness and $G = \{g_i\}_{i=1}^k$ is the family of functions $g_i$.
\end{definition}

\begin{theorem}\label{th:PCRAM_and_master}
Let $R$ be an $(f,G)$-recursive algorithm with independent branches for some $f,G$ in which the time spent on recursion level $i$ increases exponentially with $i$. Let $\mathcal{C_R}$ be a family of circuits computing the same function as $R$ such that the circuits in $\mathcal{C_R}$ can be generated by a Turing machine in time polynomial in their size. Then there exists an algorithm in the PCRAM that runs in time equal to the total time spent by $R$ on problems of size smaller than $k \eqdef \min(\text{size}^{-1}_\mathcal{C_R}(G), K)/w$, where $K$ is the greatest number such that $2K + \ceil{\log_2 K} \leq I$, plus $\text{del}_\mathcal{C_R}(kw)$. 
\end{theorem}
\begin{proof}

	Notice that $k$ is the maximum subproblem size which can be solved using the circuit.
	We run a recursion similar to $R$ while using the circuit for solving subproblems of size at most $k$. 
	The algorithm has two phases and it uses a buffer to store results from the first phase.
	
	In the \emph{first phase}, we compute results of subproblems of size at most $k$. We run $R$ with the modification that results of subproblems of size at most $k$ are computed using the circuit and stored in the buffer.
	
	In the \emph{second phase}, we compute results of subproblems of size more than $k$. We run $R$ with the modification that when subproblem size is at most $k$, we obtain the result computed in phase 1 from the buffer and return it.

\medskip
	Both phases take asymptotically same amount of time which is spent in algorithm $R$ on subproblems of size more than $k$. The second phase might have to wait for the delay of the circuit before its execution.\qed
	
\end{proof}

The memory consumption of the algorithm can be high, as we store all the results for subproblems of size at most $k$ in the buffer.	
This can be avoided by running both phases in parallel with the second phase delayed by the number of time-steps equal to the delay caused by the circuit. Then, at any moment, the number of items in the buffer does not exceed the delay of the circuit and it is, therefore, enough to use a buffer of size equal to the delay of the circuit.

\begin{corollary}
Let $R$ be a recursive algorithm with independent branches whose time complexity can be analysed by the Master theorem, running in time $t(n)$. If the time spent on recursion level $i$ of $R$ is increasing exponentially in $i$, then there is an equivalent algorithm $R'$ in the PCRAM model that runs in time $\mathcal{O}(t(n/k) + \text{del}(C_k))$ where $C_k$ is the circuit used do solve subproblems of size at most $k$.
\end{corollary}

Such problems include the Karatsuba algorithm \cite{cit:karatsuba} and the Strassen algorithm \cite{cit:strassen}. Other algorithms which are not analyzed by master theorem but can be sped up using this technique include some exponential recursive algorithms for solving $3$-SAT, maximum independent set, as well as other exponential algorithms based on the branch-and-reduce method \cite{cit:exp_survey,cit:alg_dom_set}.

Note that it may be possible to avoid using totally computing circuits by ensuring that the subproblems which are solved using the circuit have all the same size or only have constant number different sizes.

\paragraph{\textbf{Simulation of RAM or a Turing machine}} \label{simulation} can be used to get totally computing circuits. We say that a circuit of order $n$ simulates given Turing Machine (word-RAM program), if it computes the same function as the Turing machine (word-RAM program) for all inputs of size $n$.

\begin{observation}
Let $A$ be an algorithm that runs on RAM or Turing machine and computes function $f(x)$, where $x$ is the input. Then for any $n$, there is in an algorithm $A'$ which takes the first $\ceil{\log_2 n}$ inputs as the binary representation of $n'$ and computes the function $f$ on next $n'$ inputs. Moreover, there is such $A'$ that has the same asymptotic complexity as $A$.
\end{observation}

This observation can be used as a basis for creation of totally computing circuits as it is sufficient to simulate the algorithm $A'$ to get a totally computing circuit.

The following theorems say that it is possible to simulate a Turing machine or RAM on a circuit. For proof of Theorem~\ref{thm:RAM_simul}, see the appendix~\ref{app_a}. Theorem~\ref{thm:TM_simul} is well known and its proof can be found in lecture notes by \citet{cit:simulate_tm}.

\begin{theorem} \label{thm:RAM_simul}
A Word RAM with word size $w$ running in time $t(n)$ using $m(n)$ memory can be simulated by a circuit of size $\mathcal{O}(t(n)m(n)w\log(m(n)w))$ and depth $\mathcal{O}(t(n) \log (m(n)w))$.
\end{theorem}

\begin{theorem} \label{thm:TM_simul}
A Turing Machine running in time $t(n)$ using $m(n)$ memory can be simulated by a synchronous circuit of size $\mathcal{O}(t(n)m(n))$ and depth $\mathcal{O}(t(n))$.
\end{theorem}

\section{Lower Bounds}\label{section:lower_bounds}
In this section, we use the I/O model of computation to obtain several lower bounds for the PCRAM model.
\subsection{I/O model}
The computer in the I/O model is a word-RAM \cite{cit:word_ram} with two levels of memory -- cache and main memory. Both the main memory and the cache are partitioned into memory blocks. A block consists of $B$ aligned consecutive words. The processor can only work with data stored in the cache. If a word that is accessed is not in the cache, its whole block is loaded into the cache from the main memory.  When a word which is already in the cache is accessed, the computer does not access the main memory. Only $M$ words fit into the cache. When the cache is full and a block is to be loaded into it, a block which is stored in the cache has to be evicted from it. This is done according to a specified eviction strategy, usually to evict the least recently used item (abbreviated as LRU).

The \emph{I/O complexity} is defined to be the function $f$ such that $f(n)$ is equal to the maximum number of memory blocks transferred from the main memory to the cache over all inputs of length $n$.

A \emph{cache-aware algorithm} is an algorithm in the word-RAM model with the knowledge of $B$ and $M$. The measure of I/O complexity is usual for cache-aware algorithms.

\subsection{Lower bounds}
There are several lower bounds known on the I/O complexity of cache-aware algorithms. We show a general theorem which gives lower bounds in PCRAM from lower bounds in the cache-aware model. Recall that the set of circuits $\mathcal{C}$ is generated by a polynomial algorithm in word-RAM. Throughout this section, we use $t_{\text{gen}}(G, I)$ to denote the time complexity of this program.

\begin{theorem}
	Let $A$ be an algorithm in the PCRAM model running in time $t(n,G,I,w)$ where $n$ is the input size. Then there is a cache-aware algorithm $A'$ in the I/O model, simulating algorithm $A$ with I/O complexity $\mathcal{O}(t(n,M,Bw,w)+t_{\text{gen}}(M, Bw))$.
\end{theorem}
Note that since $B$ and $M$ is the number of words of cache-line and cache, respectively, they have $Bw$ and $Mw$ bits, respectively.
\begin{proof}
	We simulate algorithm $A$ using $\mathcal{O}(1)$ amortized I/O operations per instruction. We can achieve this by simulating the circuits of $A$ in the cache. Instructions of PCRAM model that do not work with circuits can be trivially performed with $\mathcal{O}(1)$ I/O operations. There exists $c<1$ such that synchronous circuits with at most $cM$ gates and $Bw$ I/O nodes can be simulated in the cache. At the beginning of the simulation, we generate the set of circuits $\mathcal{C}$ and load them into the cache. This takes time $t_{\text{gen}}(M, Bw)$. Any time we simulate circuit execution, we charge the I/O operations necessary for loading the circuit into the cache to the operation which caused eviction of the blocks which have to be loaded. This proves a bound of $\mathcal{O}(1)$ amortized I/O operation per simulated operation.

	We can, therefore, simulate $A$ with I/O complexity $\mathcal{O}(t(n, cM, Bw, w) + t_{\text{gen}}(cM, Bw))$. Since the speedup in the PCRAM model and $t_{\text{gen}}$ depend on the number of gates at most polynomially then $\mathcal{O}(t(n, cM, Bw, w)+ t_{\text{gen}}(cM, Bw)) =\mathcal{O}(t(n, M, Bw, w)+ t_{\text{gen}}(M, Bw))$. \qed
\end{proof}

\begin{corollary} \label{cor:low_bnd}
If there is a lower bound on a problem of $\Omega(t(n,M,B,w))$ in the cache-aware model, then the lower bound of $\Omega(t(n, G, I/w, w)-t_{\text{gen}}(G, I))$ holds in the PCRAM.
\end{corollary}

Some lower bounds in the cache-aware model hold even if the cache can have arbitrary content independent of the input at the start of the computation. Such lower bounds in the cache-aware model imply a lower bound of $\Omega(t(n, G, I/w, w))$ in the PCRAM. This is the case for the lower bound used in Corollary \ref{cor:predecessor_lb}.

\medskip \noindent
The following lower bound follows from Corollary~\ref{cor:low_bnd} and a lower bound shown in a survey by \citet{cit:dem_survey}.

\begin{corollary} \label{cor:predecessor_lb}
In the PCRAM model, given a set of numbers, assuming that we can only compare them, the search query takes $\Omega(\log_{I/w} n)$ time. 
\end{corollary}

More lower bounds on cache-aware algorithms are known, but, like in the Corollary~\ref{cor:predecessor_lb}, they often make further assumptions about operations that we can perform, often requiring that we treat the input items as indivisible objects. Problems for which a non-trivial lower bound is known include sorting, performing permutation and FFT \cite{cit:book_low_bnd}.

\section{Open Problems}
There are many problems waiting to be solved. An important direction of possible future work is to implement the algorithms described in this paper and compare experimental results with predictions based on time complexities in our model. The soundness of the model of computation needs to be empirically verified.

There are many theoretical problems that are not yet solved, including the following:

\begin{open}[Fast Fourtier Transform]
Is there an asymptotically fast algorithm performing the Fast Fourier Transform in the PCRAM model?
\end{open}

\begin{open}[Data structures]
Are there any data structures in the PCRAM model with asymptotically faster queries than their word-RAM counterparts?
\end{open}

\bibliographystyle{splncsnat}
\renewcommand{\bibname}{References}
\bibliography{literature}

\appendix 

\section{Simulation of Word-RAM} \label{app_a}
\setcounter{theorem}{10-1}
\begin{theorem}
A Word RAM with word size $w$ running in time $t(n)$ using $m(n)$ memory can be simulated by a circuit of size $\mathcal{O}(t(n)m(n)w\log(m(n)w))$ and depth $\mathcal{O}(t(n) \log (m(n)w))$.
\end{theorem}

\begin{proof}
We first construct an asynchronous circuit. In the proof, we will be using the following two subcircuits for writing to and reading from the RAM's memory.

\emph{Memory read subcircuit} gets as input $nw$ bits consisting of $m(n)$ words of length $w$ together with a number $k$ which fits into one word when represented in binary. It returns the $k$'th group. There is such circuit with $\mathcal{O}(m(n)w)$ gates and depth $\mathcal{O}(\log (m(n)w))$.

\emph{Memory write subcircuit} gets as input $m(n)w$ bits consisting of $m(n)$ words of length $w$ and additional numbers $k$ and $v$, both represented in binary, each fitting into a word. The circuit outputs the $m(n)w$ bits from input with the exception of the $k$'th word, which is replaced by value $v$. There is such circuit with $\mathcal{O}(m(n)w)$ gates in depth $\mathcal{O}(\log w)$.

\medskip \noindent
The circuit consists of $t(n)$ layers, each of depth $\mathcal{O}(\log (m(n)w))$. Each layer executes one step of the word-RAM. Each layer gets as input the memory of the RAM after execution of the previous instruction and the instruction pointer to the instruction which is to be executed and outputs the memory after execution of the instruction and a pointer to the next instruction. Each layer works in two phases.

In the \emph{first phase}, we retrieve from memory the values necessary for execution of the instruction (including the address where the result is to be saved, in case of indirect access). We do this using the memory read subcircuit (or possibly two of them coupled together in case of indirect addressing). This can be done since the addresses from which the program reads can be inferred from the instruction pointer.

In the \emph{second phase}, we execute all possible instruction on the values retrieved in phase 1. Note that all instructions of the word-RAM can be implemented by a circuit of depth $\mathcal{O}(\log w)$. Each instruction has an output and optional wires for outputting the next instruction (which is used only by jump and conditional jump -- all other instructions will output zeros). The correct instruction can be inferred from the instruction pointer, so we can use a binary tree to get the output from the correct instruction to specified wires. This output value is then stored in memory using the memory store subcircuit.

The first layer takes as input the input of the RAM. The last layer outputs the output of the RAM.

Every signal has to be delayed for at most $\mathcal{O}(\log(m(n)w))$ steps. The number of gates is, therefore, increased by a factor of at most $\mathcal{O}(\log(m(n)w))$. \qed
\end{proof}

\end{document}